\documentclass[11pt,a4paper]{article}
\usepackage[utf8] {inputenc}
\usepackage{amsmath,amsfonts,amssymb,amsthm}
\usepackage{xspace}
\usepackage{graphicx}
\usepackage{epstopdf}
\usepackage{subcaption}
\usepackage{enumitem}
\usepackage{float}
\usepackage{placeins}

\usepackage[parfill]{parskip}
\usepackage{authblk}

\usepackage[colorlinks,linkcolor=,citecolor=,urlcolor=blue]{hyperref}

\title{Bejeweled, Candy Crush and other Match-Three Games are (NP-)Hard}

\author[1]{L. Gualà}
\author[2]{S. Leucci}
\author[3]{E. Natale}

\affil[1]{Università degli Studi di Roma \emph{Tor Vergata} \authorcr {\tt guala@mat.uniroma2.it}}
\affil[2]{Università degli Studi dell'Aquila \authorcr {\tt stefano.leucci@univaq.it}}
\affil[3]{\emph{Sapienza} Università di Roma \authorcr {\tt natale@di.uniroma1.it}}

\newtheorem{definition}{Definition}

\newtheorem{theorem}[definition]{Theorem}

\newcommand{\NP}{\ensuremath{\textsf{\upshape NP}}}

\newcommand{\gemfont}{\fontfamily{phv}\selectfont}
\newcommand{\A}{{\gemfont{}A}\xspace}
\newcommand{\B}{{\gemfont{}B}\xspace}
\newcommand{\C}{{\gemfont{}C}\xspace}
\newcommand{\D}{{\gemfont{}D}\xspace}
\newcommand{\E}{{\gemfont{}E}\xspace}
\newcommand{\F}{{\gemfont{}F}\xspace}

\newtheorem{corollary}{Corollary}

\begin{document}
\maketitle

\begin{abstract}
The twentieth century has seen the rise of a new type of video games targeted at a mass audience of ``casual'' gamers. Many of these games require the player to swap items in order to form matches of three and are collectively known as \emph{tile-matching match-three games}.
Among these, the most influential one is arguably \emph{Bejeweled} in which the matched items (gems) pop and the above gems fall in their place.
Bejeweled has been ported to many different platforms and influenced an incredible number of similar games. Very recently one of them, named \emph{Candy Crush Saga} enjoyed a huge popularity and quickly went viral on social networks.
We generalize this kind of games by only parameterizing the size of the board, while all the other elements (such as the rules or the number of gems) remain unchanged. Then,
we prove that answering many natural questions regarding such games is actually \NP-Hard. These questions include determining if the player can reach a certain score, play for a certain number of turns, and others.
We also \href{http://candycrush.isnphard.com}{provide} a playable web-based implementation of our reduction.
\end{abstract}

\newpage

\section{Introduction}
The twentieth century has seen the rise of a video game industry targeted at a mass audience of ``casual'' gamers \cite{kuittinen2007casual}. This new market has been encouraged by the advent of web applications and the widespread diffusion of smartphones, that created a new type of ``occasional'' player looking for a gameplay experience that demands less commitment than traditional video games \cite{walsh2014candy}. As a consequence, casual games rely on very simple rules and can therefore be easily classified according to them. 

\begin{figure}[!h]
	\centerline{\includegraphics[scale=0.8]{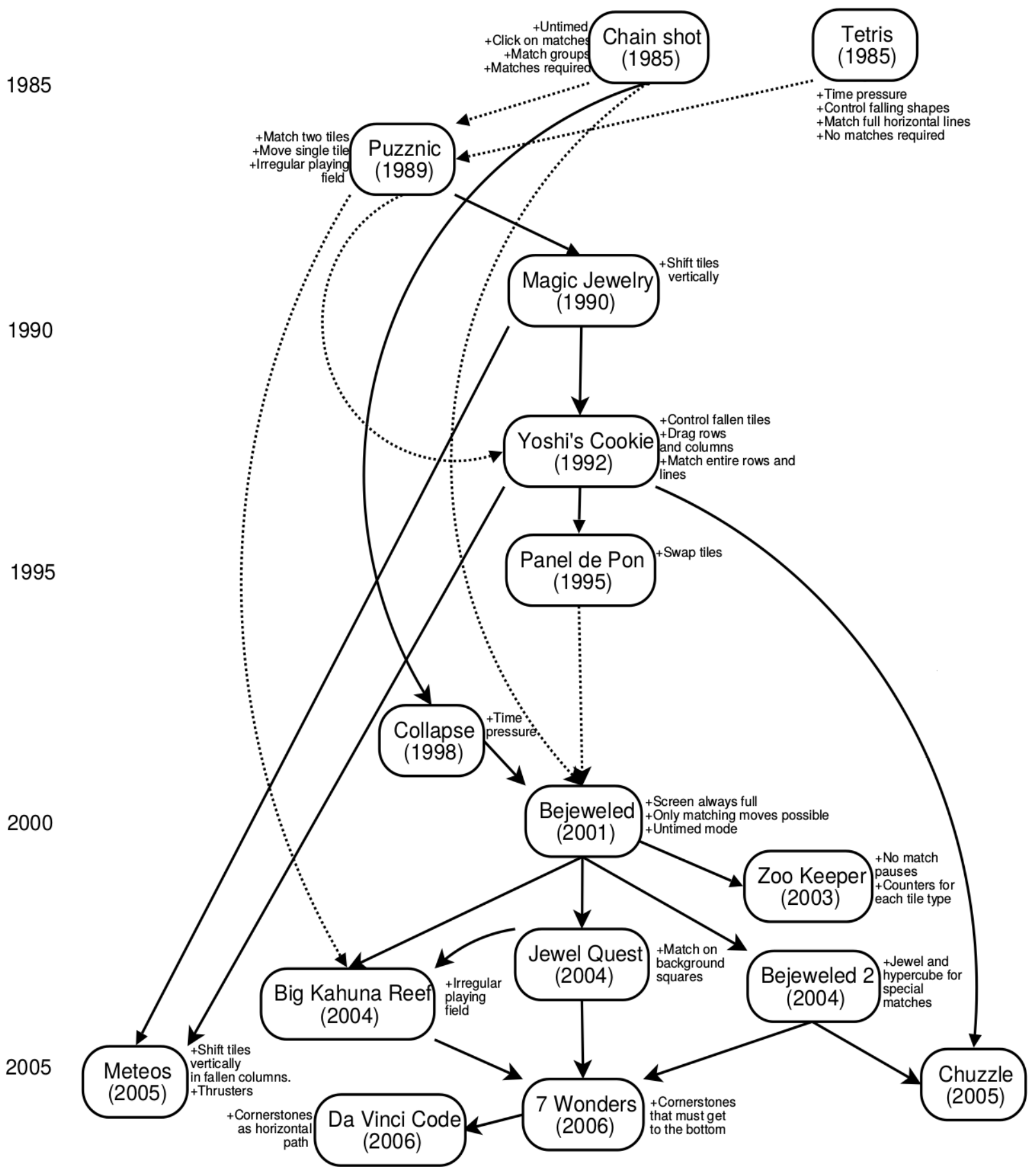}}
	\caption{The ``genealogy'' of Bejeweled\protect\footnotemark.}
	\label{fig:genealogy}
\end{figure}
\footnotetext{This is a modified image. Original image courtesy of Jesper Juul.}
\FloatBarrier

The most successful and influential game that was produced in this scenario was Bejeweled, developed for browsers by PopCap Games in 2001. It can be considered one of the most representative games of the  class of \textit{tile-matching games}, and the most important game of the class of \textit{match-three games} \cite{juul2012casual}. Among the modern games whose rules mainly recall those of Bejeweled there are AniPang, Aurora Feint, Beghouled \& Beghouled Twist, Candy Crush Saga, Diamond Twister, Gweled, Goober's Lab, Jewel Quest, King Boo's Puzzle Attack, Magic Duel, ``Puzzle Quest: Challenge of the Warlords'', Sutek's Tomb, Switch, The Treasures of Montezuma, Zoo Keeper and many others. Figure~\ref{fig:genealogy} outlines an historical perspective on the game ``ancestores'' and ``descendants'' while Figure~\ref{fig:clones} shows some screenshots of some match-three games we analyze.

The gameplay of Bejeweled is the following: 
\begin{itemize}
	\item The game is played on a board consisting of a $8 \times 8$ grid where each cell initially contains exactly one out of 6 ``\emph{gems}''.
	\item Each gem in position $(i,j)$ (row $i$ and column $j$), is considered to be adjacent to its horizontally and vertically adjacent cells, that are $\left( i-1, j \right), \left( i+1, j \right), \left( i, j-1 \right)$ and $\left( i,j+1 \right)$ (with obvious exceptions at the border of the grid).
	\item A player's move consists in swapping the positions of two adjacent gems provided that this move cause the vertical or horizontal alignment of three or more adjacent gems of the same kind. 
	\item When three or more adjacent gems of the same kind end up being vertically or horizontally aligned, they ``pop'' awarding some points to the player and disappearing from the board; the cells left empty are immediately filled with the above gems that fall towards the bottom of the grid (one can think of it as being the effect of some ``gravity''); moreover, as gems fall, the empty cells at the top of the column are filled with newly generated gems (again, one can think of them as coming from above the visible part of the grid). The whole process is repeated until there are no more gems that pop.
\end{itemize}

To summarize, the game mechanic follows this high-level algorithm:
\begin{enumerate}
	\item While there is one or more group of three or more adjacent gems of the same kind:
	\begin{enumerate}
		\item \label{rule:pop} All those gems are popped at the same time;
		\item The remaining gems fall vertically to fill the gaps and new gems are generated at the top of the board.
	\end{enumerate}

	\item If there are no allowed moves the game ends, otherwise wait for the player to make a valid move and repeat from point 1.
\end{enumerate}

 \begin{figure}[!tb]
    \centering
    \begin{subfigure}[b]{0.42\textwidth}
	  \includegraphics[scale=0.41]{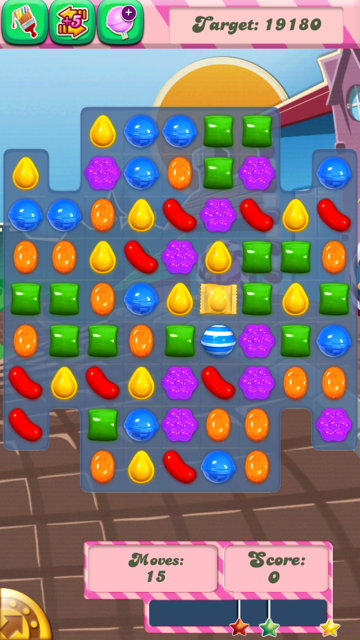}
      \caption{Candy Crush Saga}\label{fig:ccs}
    \end{subfigure}
    \hfill
    \begin{minipage}[b]{0.45\textwidth}
      \begin{subfigure}[b]{\linewidth}
	  \includegraphics[scale=0.24]{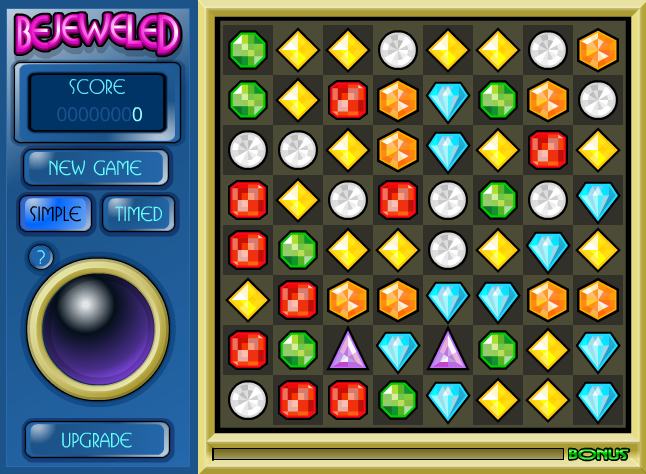}
        \caption{Bejeweled}\label{fig:bejeweled}
      \end{subfigure}\\[\baselineskip]
      \begin{subfigure}[b]{\linewidth}
	  \includegraphics[scale=0.24]{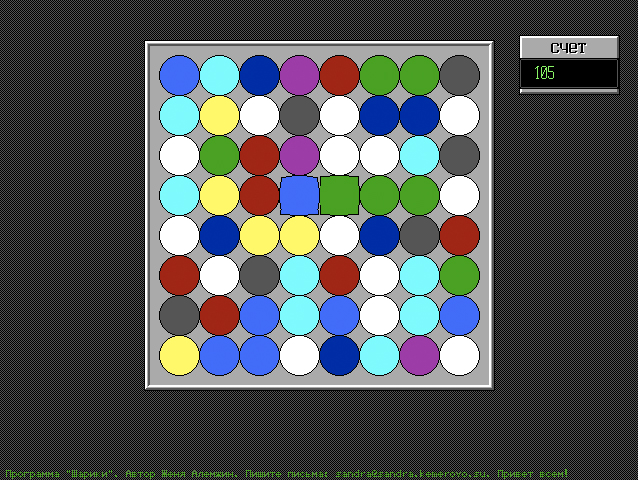}
        \caption{Shariki}\label{fig:shariki}
      \end{subfigure}
    \end{minipage}
    \caption{Screenshoots of three different tile-matching match-three games.}
    \label{fig:clones}
  \end{figure}

  Inspired by the many recent studies on the computational complexity of video games \cite{hearn2009games, demaine2003tetris, viglietta2012gaming, forivsek2010computational, kendall2008survey}, we prove that Bejeweled is an NP-hard problem. Our reduction remains valid for many of its descendant games cited above. Among them it is worth mentioning Candy Crush Saga developed by King as a Facebook and smartphones' game. This game was released in 2012 and quickly went viral, as of March 2013 it reportedly had more than 45 million average monthly users on Facebook only. 
  
\FloatBarrier
\paragraph{Our results}

We prove that Bejeweled and all similar match-three games are \NP-hard. In order to do so we first need to generalize the game and then we have to define what a decision problem associated with the game is. Intuitively, we will show that in a very natural generalization of the game it is \NP-hard to determine if the player can make a sequence of moves in order to satisfy some natural property.

In particular we consider the \emph{offline} version of Bejeweled, i.e., the game where the player actually knows the whole board and all the new gems that will appear on the top of the columns after a pop. 
The intuition is that determining if any reasonable property can always be satisfied in the \emph{online} version will not be easier than in the offline version.

Our generalization of the game is very natural: instead of playing on a board of $8\times8$ cells, we study instances of the game on a generic board of size $w\times h$; all other elements of the game remain unchanged (for example, we will only use $6$ kinds of gems as in the original game).

We mainly study the following question:
\begin{enumerate}[label=(Q\arabic*)]
	\item \label{q1} \textbf{Is there a sequence of moves that allows the player to pop a specific gem?}
\end{enumerate}
We show that answering the above question is an \NP-Complete problem.
Moreover our reduction can be easily adapted to provide the same result for the following natural questions:
\begin{enumerate}[label=(Q\arabic*),start=2]
	\item \label{q2} \textbf{Can the player get a score of at least $x$?}
	\item \label{q3} \textbf{Can the player get a score of at least $x$ in less than $k$ moves?}
	\item \label{q4} \textbf{Can the player cause at least $x$ gems to pop?}
	\item \label{q5} \textbf{Can the player play for at least $x$ turns?}
\end{enumerate}

Since our reduction is such that only matches of exactly three gems will form, our result extends to a wide class of match-three games as Shariki, Candy Crush Saga, and others.

To the best of our knowledge the only related result has been provided very recently and independently from us in \cite{walsh2014candy}, where a proof of \NP-hardness for Candy Crush was claimed. 
However, the proposed proof heavily relies on a rule different from the original one \eqref{rule:pop} given above, namely: ``In case multiple chains are formed simultaneously, we assume that chains are deleted from the bottom of the board to the top as they appear, and the candies above immediately drop down''. 
This assumption makes the gameplay quite different from that of Bejeweled/Candy Crush. Just to give an example, we can think of a situation where a swap leads to a simultaneous horizontal alignments of gems on several different rows such that the following holds: within the original game this gems pop at the same time, if instead the cited rule is assumed we would have that, after the gems in the bottom row pop, every other alignment is broken resulting in a configuration where no more matches are left.

Our result is also stronger in the following sense: the reduction provided in~\cite{walsh2014candy} shows that it is \NP-Hard to determine if a certain score can be obtained in a given number of moves (as in \ref{q3}). However, in their instances of the game, if the constraint on the number of moves is removed, then it is easy to find (in polynomial time) a strategy that achieves the given score. This is quite undesirable since, while in Candy Crush there is a limit on the number of player's moves, no such constraint exists in Bejeweled and in other games.
On the contrary, this problem does not arise in our reduction since we consider problems where the number of moves is unconstrained. 

Finally, we also provide a (playable!) web-based implementation of our reduction, released under GPL license, which is accessible at the URL \linebreak \url{http://candycrush.isnphard.com}.

\section{Reduction overview}

The reduction is from \emph{1-in-3 positive 3SAT} (in short 1in3PSAT), which is known to be NP-hard \cite{garey1979computers}. In this problem, we are given $n$ variables $x_1,\dots,x_n$, $m$ clauses, each having at most 3 variables, and we want to choose a subset of variables to set to true in such a way that every clause has exactly one true variable. We call such a truth assignment a \emph{satisfying assignment}. Given an instance of 1in3PSAT, we build an instance of Bejeweled where a given goal gem can be popped if and only if the corresponding 1in3PSAT instance is satisfiable. We will argue later how this reduction can be simply modified to show that answering to any of the questions discussed above is also \NP-Complete.

We now provide an intuitive description of the reduction which uses a set of high-level gadgets. Then, we will discuss how to implement each gadget. The whole scheme of the reduction in given in Figure \ref{fig:overview}. To each variable corresponds certain adjacent rows (which contains certain \emph{wire gadgets} and \emph{displacer gadgets}), while each clause corresponds to three adjacent columns.

We have also a gadget called \emph{sequencer}. We use the sequencer to ensure that the control of the player follows a prescribed way. At the beginning there are only two feasible moves at the \emph{start point}, which intuitively coincides with the choice of $x_1$. Then, the control will be moved down through all the \emph{wires} of $x_1$'s gadget until it reaches the $x_2$'s gadget, and so on. Once the gadget of the last variable has been traversed, the control is given to the \emph{check point} in the \emph{goal area} where a specified \emph{goal wire} can be taken if and only if the previous choices have made the formula satisfied.

\begin{figure}[!t]
	\centerline{\includegraphics{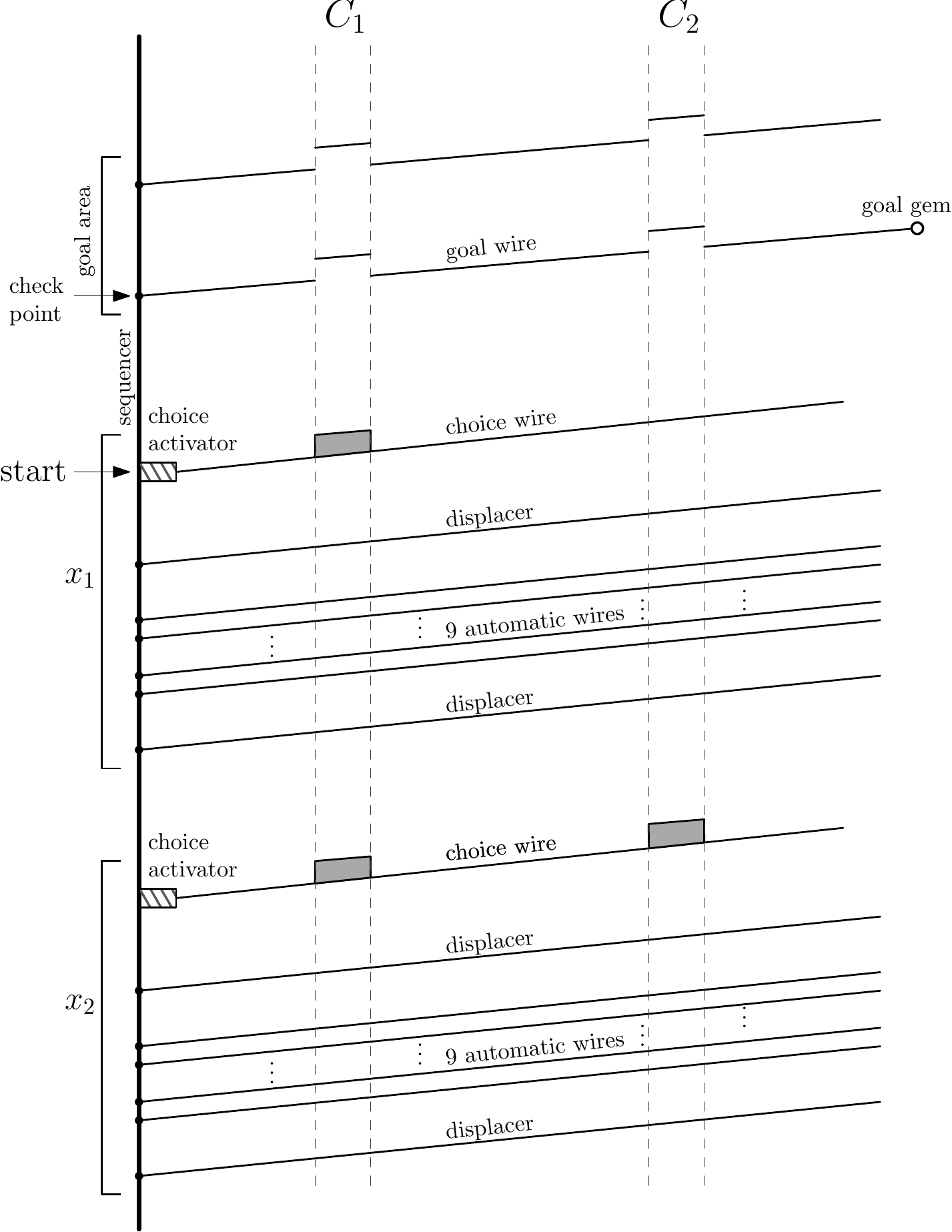}}
	\caption{Overview of the structure of the generated instance.}
	\label{fig:overview}
\end{figure} 

Let us describe the $x_i$'s gadgets. We have one \emph{choice wire}, a first displacer, nine \emph{automatic wires}, and a final displacer. The choice wire can be either activated or skipped, which intuitively corresponds to setting $x_i$ to true or false. Activating or skipping $x_i$'s choice wire causes the columns to fall by a different number of cells. More precisely, if we set $x_i$ to false, for every column the number of popped gems will be a multiple of 6. Otherwise, if we set $x_i$ to true, the columns corresponding to the clauses which contain $x_i$ fall by some number $\ell \equiv 2 \pmod{6}$. In order to prevent that the already made choice to be changed, a suitable gadget named \emph{shredder} will destroy the choice wire by making the columns in odd and even positions fall by respectively different amounts. Once we have traversed all the gadgets of the variables, the control is given to the check point in the goal area.

The goal area is a set of copies of a goal wire, one copy each six rows. Intuitively, a goal wire is a wire which is not aligned in correspondence to the columns of the clauses by a gap of two cells. If we have found a satisfying assignment of the initial 1in3PSAT instance, every goal wire gets aligned and there is one of them which allows to reach the goal gem. On the other hand, any assignment which does not satisfy the initial instance will not make any goal wire to be aligned. Indeed, satisfying none or more than one variable in a clause will result in falling by either a number $\ell \equiv 0 \pmod{6}$ or  $\ell \equiv 4 \pmod{6}$ of cells. This would imply that no goal wire can be activated.

Finally, in order to prevent unwanted interactions, we leave large gaps between the gadgets of one variable and that of the next one. We also leave a large gap between the end of goal zone and the gadgets of the first variable.

\FloatBarrier
\section{Reduction details}

In this section we describe the reduction in all details.

\subsection{Notation}

We will refer to the six different kinds of gems using the letters \A,\B,\C,\D,\E, and \F. Moreover, we will also use a single dot to represent some kind of special gems. Dot-gems never interact with each other nor with other gems, so they cannot pop, but they still fall due to gravity.
We stress that this special dot-gems are not needed at all for the reduction to work as they will only be used in order to present the used gadgets in a clear way, without being distracted by unnecessary gems and interactions.

After discussing all the gadgets we will show how these gems can be safely removed by using a certain  pattern to replace them.

\subsection{Sequencer}

The sequencer is a crucial gadget in our reduction. This gadget controls the order in which the other gadgets are activated. Roughly speaking, when the sequencer is ``running'' a sequence of pops of three gems is taking place (therefore the player cannot make any move) in order to move the control from a certain gadget to the next one.
We have two kinds of gadgets: automatic and manual ones. Automatic gadgets will trigger some effect as soon as they get the control from the sequencer, namely they get activated. However, they will not stop the sequencer which keeps running towards the next gadget.

On the other hand, manual gadgets pause the sequencer.
In this phase, the player can decide either to interact with the manual gadget or to skip it. In any case a certain swap is needed in order to set the sequencer back in motion.

The sequencer is positioned on the first column of the field and we can decompose its gems into two kinds of logical blocks: the first kind is composed of the gems needed to produce a chain of pops, as described before, while the second kind of blocks are the so-called ``sequencer records''.
A sequencer record is simply a sequence of gems that will align with the next gadget to be activated. When the record is in place, its gems along with the gems on the following two columns of the gadgets will interact (and some will pop) and this will effectively activate the gadget. After the control moves out of the gadget, the sequencer record will happen to be completely destroyed.

Please note that each gadget has its own sequencer record and that the record for different kinds of gadget are composed by different sequences of gems. In the following pictures, the sequencer record for the described gadgets will be found on the first column, already aligned, and its gems will be highlighted in red. 

We now describe how to construct a sequencer for a given set of gadgets.
We recall that, except for the goal wire, the gadgets have to be activated from top to bottom. Consider the first column to be initially empty, we start by placing the sequencer record for the first gadget, i.e. the choice wire corresponding to $x_1$, so that it is already aligned. This will cause the control to be immediately transferred to the first choice wire. Next we add the gems needed to reach the next gadget: we place a \C on top of the pile on the first column and two \C{}s underneath, then we perform the same operation using a \D, and we continue repeating the above, alternating \C{}s and \D{}s until the last gem on the bottom reaches the second gadget. In particular we want this gem to lie in the same place where the bottom gem of the sequencer record (of the second gadget) will have to be placed. At this point we place the sequencer record on the top of the pile, and repeat the whole procedure for each remaining gadget.

\begin{figure}[!t]
	\centerline{\includegraphics{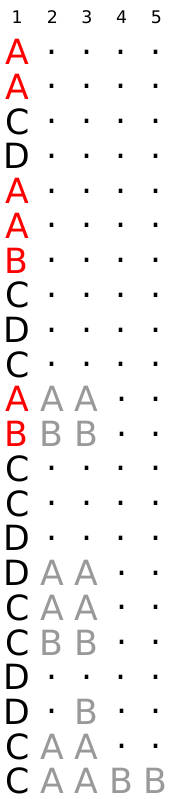}}
	\caption{Example of a sequencer.}
	\label{fig:sequencer}
\end{figure} 

An example is given in Figure~\ref{fig:sequencer}, where three dummy gadgets are shown in gray.\footnote{These gadgets do not serve any purpose and will not be used in the reduction. They are only used in order to explain how the sequencer works.} These gadgets will simply cause all their gems to pop when activated. Their corresponding activation records are shown in red while the other gems of the sequencer are shown in black.

The last gadget to activate is the goal wire. As this is placed above all the other gadgets, the construction we just described needs to be extended.
It is sufficient to make the sequencer fall by some additional large amount\footnote{This can easily be done by adding alternating patterns of \C{}s and \D{}s to the top and the bottom of the already built sequencer, as shown before.} after the activation of all the other gadgets, then we place the activation record for the goal wire (which will be a single \A{}) high enough on the first column in order to make it align with the goal wire after the large fall stops. 

As a final remark, we point out that in our instances all the falls will happen in multiples of three gems. Moreover, all the activation records we will use are designed so that, if a record is still falling towards its gadget, no unwanted matches can form between its gems and the adjacent gems of the gadget. 
 
\FloatBarrier
\subsection{Choice wire}
\label{sec:choice_wire}

\begin{figure}[!t]
	\centerline{\includegraphics{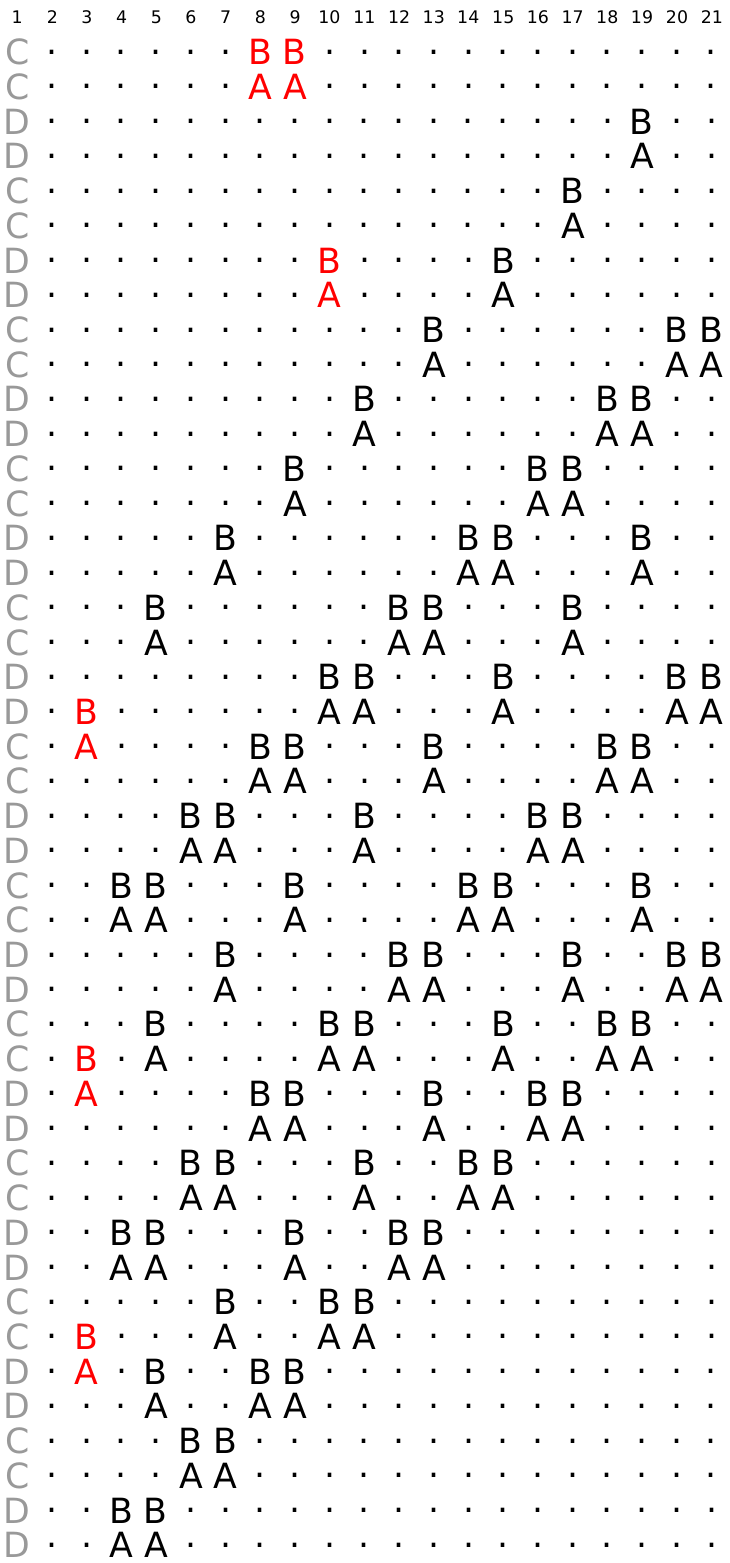}}
	\caption{Choice wire gadget.}
	\label{fig:choicewire}
\end{figure} 

The choice wire associated to a certain variable $x_i$ is a gadget which will be activated if and only if $x_i$ is set to true. 
Once the wire is activated (see the choice activator gadget) it ensures that the columns corresponding to the clauses containing $x_i$ fall by a number $\ell \equiv 2 \pmod{6}$ while other columns fall by multiples of 6. 

An example of a wire is shown in Figure~\ref{fig:choicewire}. The wire is activated when the third column falls by exactly $6$ cells, aligning the red pairs of gems to the black gems of the fourth and fifth columns, forming 6 matches at the same time. After the gems of these matches pop, the gems of the fifth column fall to fill the gaps and form a match with the gems in the next two columns.
This construction is repeated until the end of the wire, which is slightly different in order to ensure that all the gems pop.

The first group of three columns that correspond to a clause starts from column $8$ (and ends at column $10$). Every other group starts after $6$ columns from the start of the previous group (i.e., the second group start at column $14$ and ends at column $16$ and so on).

If $x_i$ belongs to a certain clause $c_j$, then two additional gems are placed on the columns corresponding to $c_j$. The gems corresponding to the first clause are shown in red in Figure~\ref{fig:choicewire}. When the wire is activated these gems will align, therefore the corresponding columns (from $2+j\cdot 6$ to $4+j\cdot 6$) will fall by an additional $2$ cells. 

Notice how, except for these gems, every other column falls by either $6$ or $12$ gems, in an alternating fashion.

\FloatBarrier
\subsection{Choice activator}

\begin{figure}[!t]
	\centerline{\includegraphics{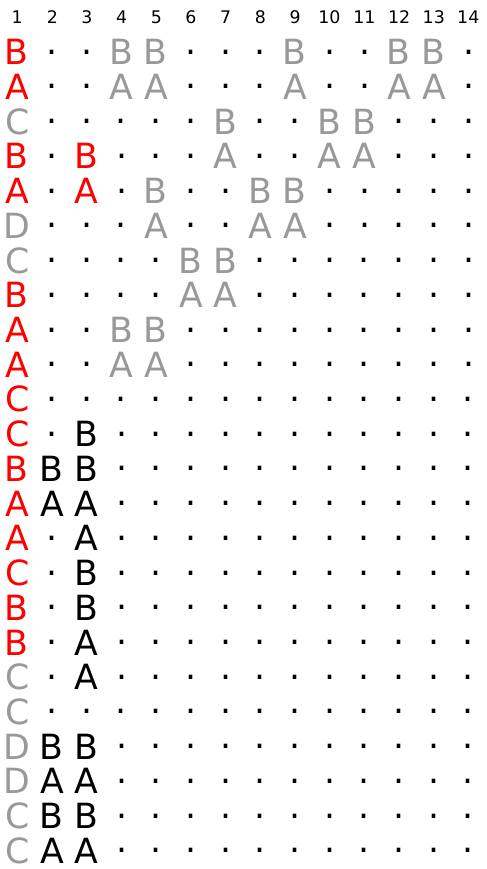}}
	\caption{Choice activator gadget.}
	\label{fig:manactiv}
\end{figure} 

The choice activator, associated with a variable $x_i$ is attached to both the sequencer and the choice wire for $x_i$. This gadget can either be activated by the player, effectively setting $x_i$ to true, or it can be skipped, effectively setting $x_i$ to false. An example of this gadget is shown in Figure~\ref{fig:manactiv}.

If the player chooses to activate the gadget then the third column will fall in such a way that it will activate the corresponding wire (shown in gray), otherwise no activation will occur.

We now the describe how such a gadget can be played by the player.
As soon as the sequencer record (shown in red) aligns with the gadget, two horizontal matches are formed and the corresponding gems are popped. Now the player has two available moves:
\begin{itemize}
	\item Swapping a \B on the third column with the \A underneath, causing three \B{}s to pop and then three \A{}s to pop. This also activates the wire as described before. This corresponds to setting $x_i$ to true.
	\item Swapping a \C on the first column with the \A beneath, causing the sequencer to continue it's operation. This will make the first two columns fall by $4$ cells and the control to be transferred to the next gadget. This corresponds to skipping the activation of the choice wire, effectively setting $x_i$ to false.
\end{itemize}

If the players chooses the first move then, after the wire has completely popped, he will also have to make the second move (which is the only one remaining).
If the player choses the second move this will cause the sequencer to (automatically) activate the next set of gadgets which will ``destroy'' the wire. As a consequence, the next time the player will have chance to swap a gem, it will no longer be possible to perform the first move.\footnote{Actually it will still be possible to pop the gems on the third column, but this will not produce any useful result.}

\FloatBarrier
\subsection{Shredder}

Consider the choice activator gadget for $x_i$, after the player swaps the \C in the first column, the control moves to the next gadget. Recall that this happens regardless of the choice made about $x_i$.
If the player has set $x_i$ to false we use a ``shredder'' gadget in order to ``destroy'' the wire so that it cannot be activated anymore.
Otherwise, $x_i$ is set to true and all the gems of the choose wire have already been popped. In this case the shredder will still be activated but, clearly, it won't produce any effect.

The shredder will work by moving odd and even columns far from each other.
We need, however, to be careful to avoid any fortuitous match of three gems of the previous choice wire (that has voluntarily been skipped).
Indeed if one column falls by a multiple of $2$ cells, three or more gems of the choice wire gadget might align. 
We solve this problem by breaking the alignment of those columns in such a way that any fall by an even number of cells cannot trigger spontaneous pops. This is achieved by the use of a displacer gadget (described in the following).

\begin{figure}[!t]
	\centerline{\includegraphics{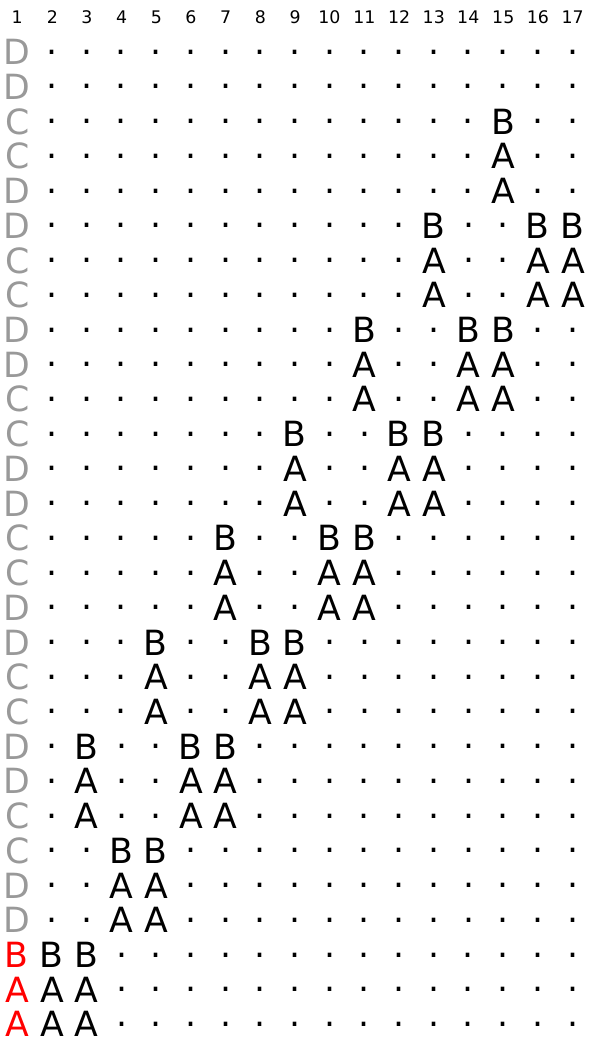}}
	\caption{Displacer gadget.}
	\label{fig:displacer}
\end{figure}

After the displacer we place a number of automatic wires (described in the following) that will make odd and even columns of the wire fall by even and largely different amounts. This will effectively prevent any further pop among the gems of the choice wire.

Finally, we place another displacer in order to guarantee that the total number of gems popped in each column during the whole process will be a multiple of $6$.

We now describe the displacer and the automatic wire gadgets.

\subsubsection{Displacer}

The displacer gadget is shown in Figure~\ref{fig:displacer} and its construction is quite simple.
Once activated with the proper sequencer record, a chain of pops will ensure that even columns fall by $3$ cells while odd columns fall by $6$ cells.\footnote{Except the last column, that will fall by $3$.}

Moreover while these pops occur, it will never happen that three or more gems align in the preceding choice wire.

\FloatBarrier
\subsubsection{Automatic wire (and activator)}

\begin{figure}[!t]
	\centerline{\includegraphics{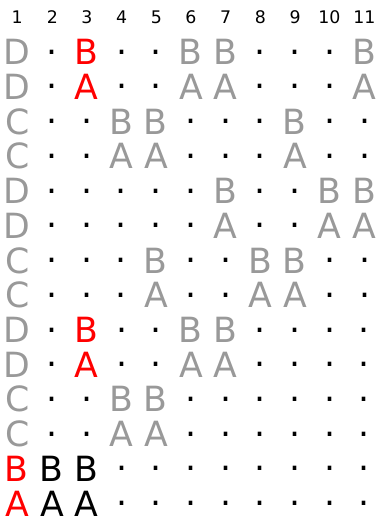}}
	\caption{Automatic activator gadget. The wire construction is similar to the one of the choice wire gadget and its initial portion is shown in gray.}
	\label{fig:autoactiv}
\end{figure}
 
An automatic wire is very similar to the choice wire described above, except it has no additional gems on the columns corresponding to the clauses and it is preceded by a so called ``automatic activator'' instead of the already described choice activator. 

The automatic activator is a very simple gadget that has the sole purpose of automatically activating a wire whenever its sequencer record falls in place. The design of an automatic activator is shown in Figure~\ref{fig:autoactiv}.

\FloatBarrier
\subsection{Goal wire}
\label{sec:goal_wire}

Once we have traversed all the variable gadgets, the sequencer gives the control to goal zone by placing an ``activating gem'' at the check point. 
The task of the goal zone is to ensure that the truth assignment resulting from the variable gadget activations is a satisfying assignment. If this is the case, it will result in the possibility to reach the goal gem from the check point.

\begin{figure}[!t]
	\centerline{\includegraphics{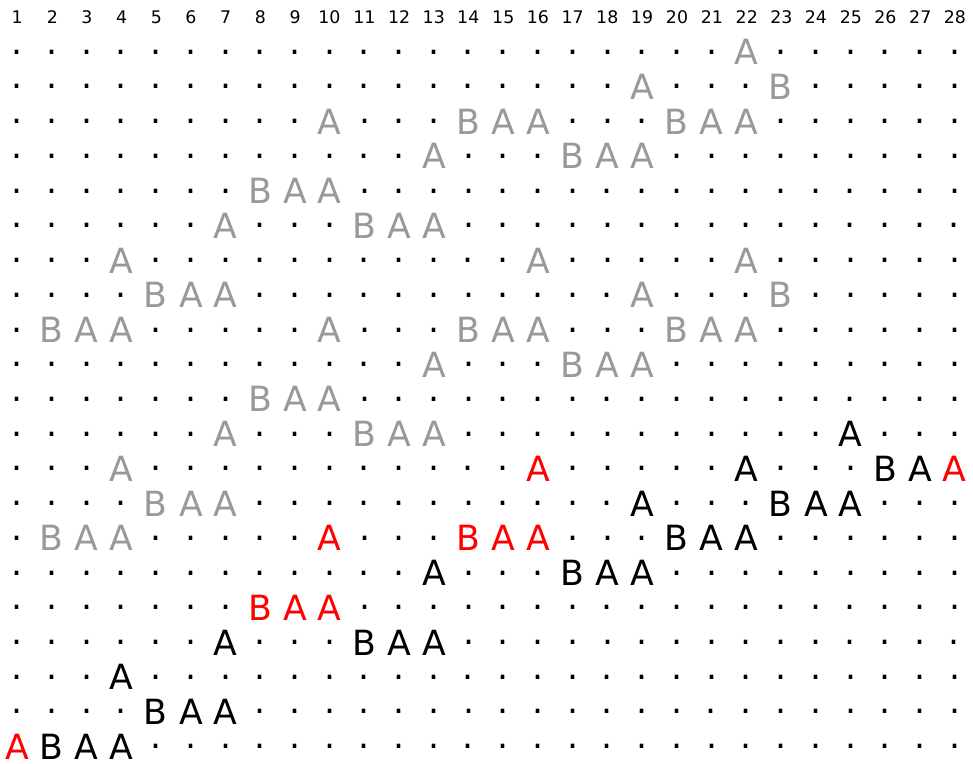}}
	\caption{Goal wire gadget in its original configuration. There has two clauses that are not satisfied so the wire cannot be traversed.}
	\label{fig:goalwire}
\end{figure} 

The goal zone essentially consists of a suitable number of copies of a gadget which we call goal wire. We have one copy each six row, with the bottom copy ending with an additional structure that leads to the goal gem, as shown in Figure~\ref{fig:overview} and Figure~\ref{fig:goalwire}. 
The reason for having a goal wire each six rows is the following: despite we don't know in advance by which amount the columns below the goal zone will fall, the underneath gadgets guarantees that in each column the number of popped gems is a multiple of six, except for the columns corresponding to the clauses. Indeed, each column corresponding to a given clause $c_i$ has fallen by $k\cdot2\equiv 0\pmod 6$, where $k$ is the number of variables of $c_i$ whose gadget has been activated. Since each clause has at most three variables then $0\leq k \leq 3$). If no variable of $c_i$ has been set to true, $k=0$ and also these columns has fallen by a multiple of six (this situation is depicted in Figure \ref{fig:goalwire}). On the other hand, if $k=1$ the gems of the goal wire in the columns of $c_i$ get aligned with the rest of the wire (see Figure \ref{fig:goalwire2}), allowing to reach the goal gem by a sequence of obvious swaps. 
Finally, note that if $k=2$ or $k=3$, the goal wire does not get aligned.

\begin{figure}[!t]
	\centerline{\includegraphics{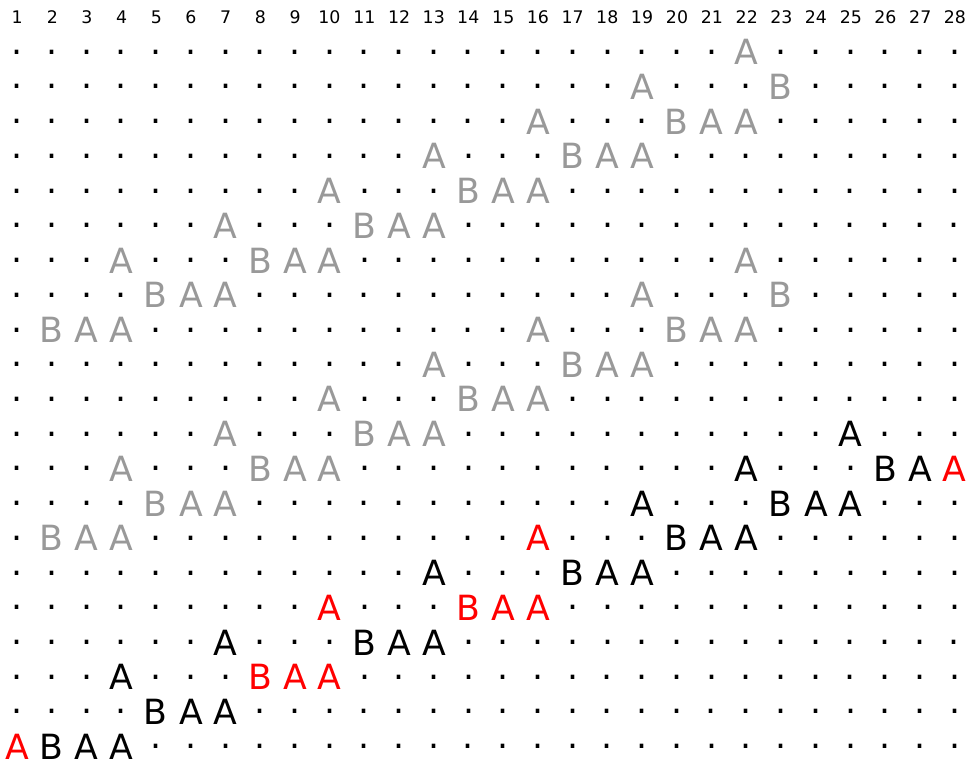}}
	\caption{Goal wire gadget where both clauses have been satisfied. The wire can be traversed to reach the goal gem (the red gem on column 28).}
	\label{fig:goalwire2}
\end{figure}

In order to activate the goal wire, the player can swap the \A on the first column with the \B on the second one. This will cause the now formed match of \A{}s to pop and allows the whole procedure to be repeated until the goal gem is reached.
While, after the first move, it is possible to also activate some goal wires that are placed above the one that gets aligned with the goal gem, doing so will not be of any help in reaching the goal gem.

\FloatBarrier
\subsection{Removing the dot-gems}

Finally, we discuss how the dot-gems we used in the previous gadgets can be safely replaced with other gems.

Consider the red pattern shown in Figure~\ref{fig:filler}, and notice that the player cannot swap any gem to form a match, even if the columns fall by different relative amounts.

As in all the columns but the first we only used \A{}s and \B{}s so far, we can safely replace all the dot-gems of those columns following the above pattern. We only consider dot-gems: when another gem is encountered it is simply ignored and we proceed to the next dot-gem. To clarify, we consider the second column and we move among its gems from top to bottom, every time we find a dot-gem we replace it with an \E or \F, in an alternating fashion. Then, we move to the next column where we use the gems \C and \D. The following column will again use \E and \F, and so on.

\begin{figure}[!t]
	\centerline{\includegraphics{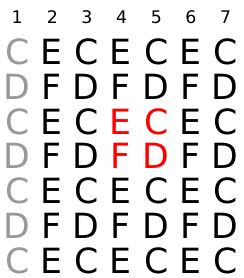}}
	\caption{Alternating pattern of gems where no swaps can be made.}
	\label{fig:filler}
\end{figure} 

Now, we only need to handle the first column. In order to do so we consider all the dot-gems placed in the first column above the sequencer record for the first gadget. Since in this portion of the column  the dot-gems can only be adjacent to \A{}s, \B{}s or \emph{alternating} sequences of \C{}s and \D{}s, it suffices to replace these dot-gems with \C{}s and \D{}s, following the already existent alternating fashion.
%Regarding the dot-gems placed below the first sequencer (i.e. after the end of the sequencer), we first place two \B{}s just after the end of the sequencer.
%Notice how this cannot allow the formation of any match as the last gems of the sequencer are alternating \C{}s and \D{}s and the last gadget in the field is a displacer.
%After these two \B{}s, we continue replacing any remaining dot-gem with the usual pattern of \C{}s and \D{}s.

When new gems are needed at the top of the board, they are generated according to the above pattern.

\section{Putting all together}

\begin{theorem}
	Given an instance of Bejeweled, the problem \ref{q1} of deciding whether there exists a sequence of swaps that allows a specific gem to be popped is \NP-Hard. 
\end{theorem}
\begin{proof}
	Let $\phi$ be a formula for the 1in3PSAT problem and construct the Bejeweled field as shown above.
	We consider any possible sequence of moves from the beginning of the game until the sequencer record of the goal wire reaches the check-point. We show how any such sequence can be mapped to a corresponding truth assignment $\pi$ for the variables $x_1, \dots, x_n$, and then we argue that the goal gem can be popped if and only if $\pi$ satisfies $\phi$.
	
	Now we argue that the freedom of the player is essentially restricted in choosing whether to activate or skip the choice wires. Moreover, due to the sequencer, these choices have to be performed in order.
	Indeed, at the beginning of the game the player only has two available moves: he can either activate the choice wire of $x_1$ and then the sequencer, or skip it by only activating the sequencer. 
	In the latter case, the player will still be able to swap a \B of the corresponding choice activator (see column 3 of Figure~\ref{fig:manactiv}) but doing so will only cause $6$ gems in the same column to pop, without affecting any other gadget. We consider $x_1$ to be true if and only if the choice wire has been activated. Notice that, as discussed in Section~\ref{sec:choice_wire}, all the columns corresponding to the clauses containing $x_1$ will fall by a number $\ell \equiv 2 \pmod{6}$ of cells.
	The same argument applies to the following choice wires. When all the choice wires have been activated or skipped the sequencer record for the goal wire reaches the check-point.
	
	As discussed in Section~\ref{sec:goal_wire}, we have that the goal wire can only be traversed to reach the goal gem if and only if the columns of all the clauses have fallen by some number $\ell \equiv 2 \pmod{6}$ of cells. This happens if and only if exactly one variable per clause has been set to true.
\end{proof}

\begin{figure}[!t]
	\centerline{\includegraphics{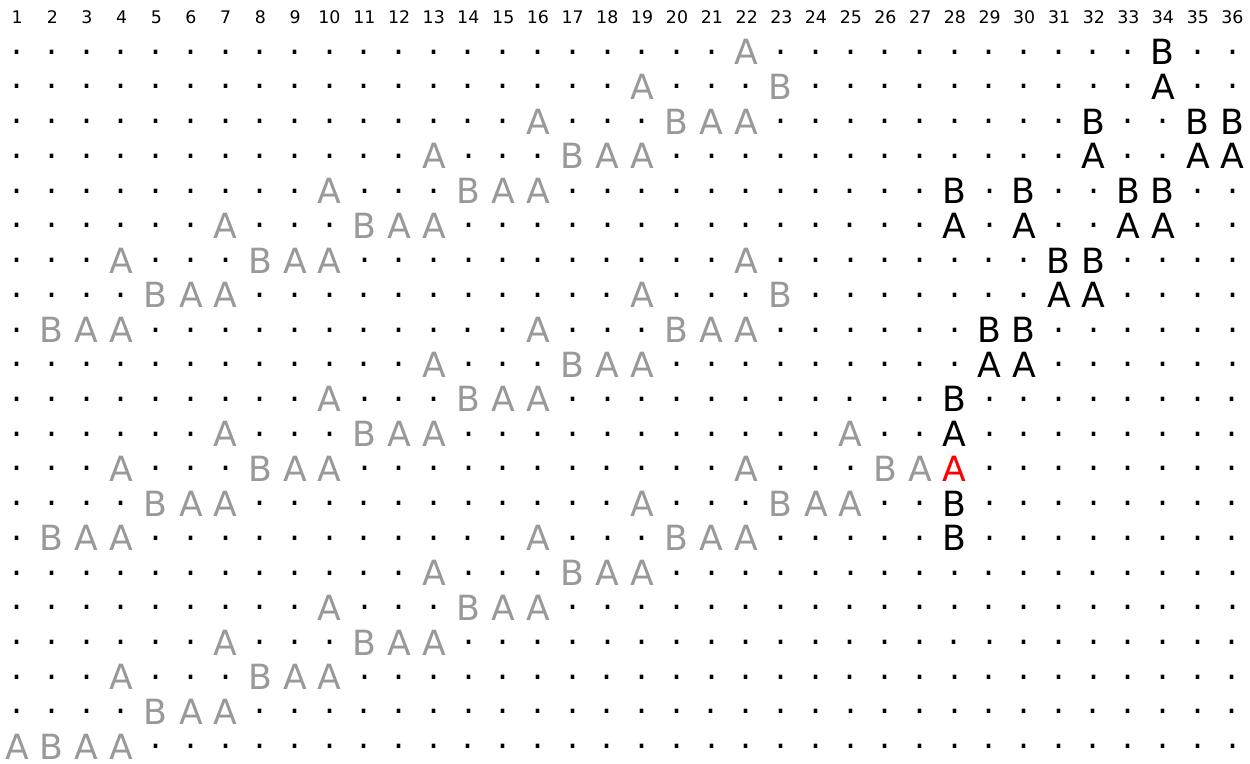}}
	\caption{Wire used in the reduction for \ref{q2}.}
	\label{fig:goal_wire_q2}
\end{figure} 

As a consequence of the previous theorem one can easily prove the following result.
\begin{corollary}
	The decision problems \ref{q2}, \ref{q3}, \ref{q4}, and \ref{q5} are NP-Hard.
\end{corollary}
\begin{proof}
	We now show how the given construction can be adapted to address each decision problem:
	
	\begin{description}
		\item[\ref{q2} and \ref{q4}] We can add a new kind of wire, which can be activated if and only if the goal gem is popped, as shown in Figure~\ref{fig:goal_wire_q2}. This wire can be arbitrary long and, once activated, will ensure that all its gems will pop in a single continuous sequence. For a suitable length of the wire and choice of $x$, it is obvious that the player can reach a score of $x$ (or pop at least $x$ gems) if and only if the goal gem is popped.
		
		\item[\ref{q3}] This is immediately implied by the hardness of \ref{q2} once we set $k$ to a large (polynomial) value.
		
		\item[\ref{q5}] We can modify the instance by suitably extending the goal wire after the goal gem thus allowing the player to play a given number of moves if and only if the goal gem is popped. 
	\end{description}
	
\end{proof}

\newpage
\bibliographystyle{plain}
\bibliography{main}

\begin{thebibliography}{1}

\bibitem{demaine2003tetris}
Erik~D Demaine, Susan Hohenberger, and David Liben-Nowell.
\newblock Tetris is hard, even to approximate.
\newblock In {\em Computing and Combinatorics}, pages 351--363. Springer, 2003.

\bibitem{forivsek2010computational}
Michal Fori{\v{s}}ek.
\newblock Computational complexity of two-dimensional platform games.
\newblock In {\em Fun with Algorithms}, pages 214--227. Springer, 2010.

\bibitem{garey1979computers}
Michael~R Garey and David~S Johnson.
\newblock {\em Computers and intractability}, volume 174.
\newblock freeman San Francisco, 1979.

\bibitem{hearn2009games}
Robert~A Hearn and Erik~D Demaine.
\newblock {\em Games, puzzles, and computation}.
\newblock AK Peters Wellesley, 2009.

\bibitem{juul2012casual}
Jesper Juul.
\newblock {\em A casual revolution: Reinventing video games and their players}.
\newblock The MIT Press, 2012.

\bibitem{kendall2008survey}
Graham Kendall, Andrew~J Parkes, and Kristian Spoerer.
\newblock A survey of np-complete puzzles.
\newblock {\em ICGA Journal}, 31(1):13--34, 2008.

\bibitem{kuittinen2007casual}
Jussi Kuittinen, Annakaisa Kultima, Johannes Niemel{\"a}, and Janne
  Paavilainen.
\newblock Casual games discussion.
\newblock In {\em Proceedings of the 2007 conference on Future Play}, pages
  105--112. ACM, 2007.

\bibitem{viglietta2012gaming}
Giovanni Viglietta.
\newblock Gaming is a hard job, but someone has to do it!
\newblock In {\em Fun with Algorithms}, pages 357--367. Springer, 2012.

\bibitem{walsh2014candy}
Toby Walsh.
\newblock Candy crush is np-hard.
\newblock {\em arXiv preprint arXiv:1403.1911}, 2014.

\end{thebibliography}

\end{document}